\documentclass{proc-l}
\usepackage{amssymb}

\newtheorem{theorem}{Theorem}

\newtheorem{corollary}{Corollary}

\theoremstyle{remark}
\newtheorem{remark}{Remark}

\newenvironment{acknowledgement}{\par\medskip\noindent\emph{Acknowledgment.}}

\renewcommand\d{\mathrm d}

\def\le{\leqslant}
\def\ge{\geqslant}

\def\Chi{\raisebox{.4ex}{$\chi$}}

\DeclareMathOperator{\tr}{tr}

\DeclareMathOperator{\supp}{supp}

\def\esssup{\mathop\textup{ess\,sup}}
\def\essinf{\mathop\textup{ess\,inf}}

\newcommand{\R}{\mathbb{R}}

\newcommand{\Z}{\mathbb{Z}}

\makeatletter

\newif\ifper\pertrue
\def\per{.}

\newcounter{aucount}
\def\HarvardComma{,}

\newif\ifedplural

\def\au#1#2{{#1 #2}\addtocounter{aucount}{1}}
\def\lau#1#2{{#1 #2},\setcounter{aucount}{0}}

\def\ed#1#2{\ifnum\theaucount=0(\fi{#1 #2}\addtocounter{aucount}{1}}
\def\led#1#2{\ifnum\theaucount=0(\edpluralfalse\else\edpluraltrue\fi{#1
    #2} (\editorname.)):\setcounter{aucount}{0}}
\def\editorname{\ifedplural Eds\else Ed\fi}

\def\et{\ifnum\theaucount=1\else\HarvardComma\fi{} and\ }
\def\ti#1{\emph{#1},}

\def\bti{\@ifnextchar[\bbti\bbbti}
\def\bbti[#1]#2{\emph{#2}, #1,}
\def\bbbti#1{\emph{#1},}

\def\z{\@ifnextchar[\zz\zzz}
\def\zz[#1]#2#3#4#5{\perfalse{#2} \textbf{#3} (#5), #4 [#1]\per}
\def\zzz#1#2#3#4{{#1} \textbf{#2} (#4), #3\ifper\per\fi\pertrue}

\def\pub{\@ifstar\pubstar\pubnostar}
\def\pubnostar{\@ifnextchar[\@@pubnostar\@pubnostar}
\def\@@pubnostar[#1]#2#3#4{#1, #2, #3, #4\per}
\def\@pubnostar#1#2#3{#1, #2, #3\ifper\per\fi\pertrue}
\def\pubstar[#1]#2#3#4{\perfalse #2, #3, #4 [#1]\per}

\makeatother

\makeatletter
\AtBeginDocument{\def\ps@firstpage{\ps@plain
        \def\@oddhead{{\scriptsize
            \begin{tabular}{@{}l@{$\;$}l@{$\;$}l@{$\;$}l@{}r@{}}
            To & appear & in & the & PROCEEDINGS$\;\,$OF$\;\,$THE\\[0.2ex]
            \multicolumn{5}{@{}l@{}}{AMERICAN$\;\,$MATHEMATICAL$\;\,$SOCIETY}
            \end{tabular}
          }\hss}
        \def\@oddfoot{}%
      }}
\makeatother

\begin{document}

\title[A lower bound for the density of states]{A lower bound for the
density of states of the lattice Anderson model}

\author[P.\ D.\ Hislop]{Peter D.\ Hislop}
\address{Department of Mathematics,
    University of Kentucky,
    Lexington, Kentucky  40506-0027, USA}
\email{hislop@ms.uky.edu}

\author[P.\ M\"uller]{Peter M\"uller}
\address{Institut f\"ur Theoretische Physik,
  Georg-August-Universit\"at G\"ottingen,
  Friedrich-Hund-Platz 1,
  37077 G\"ottingen, Germany}
\email{peter.mueller@physik.uni-goe.de}

\copyrightinfo{}{American Mathematical Society}
\date{}

\dedicatory{Dedicated to Jean-Michel Combes on the occasion of his
  65$^{\mbox{th}}$ birthday}


\begin{abstract}
  We consider the Anderson model on the multi-dimensional cubic
  lattice and prove a positive lower bound on the density of states
  under certain conditions.  For example, if the random variables are
  independently and identically distributed and the probability
  measure has a bounded Lebesgue density with compact support, and if this
  density is essentially bounded away from zero on its
  support, then we prove that the density of states is strictly
  positive for Lebesgue-almost every energy in the deterministic
  spectrum.
\end{abstract}

\maketitle


Wegner's estimate, originally formulated in \cite{Weg81} for the
Anderson model on the lattice $\mathbb{Z}^{d}$, is one of the
celebrated tools in the theory of random Schr\"odinger operators, see
e.g.\ the recent reviews \cite{Ves04, KiMe07} or \cite{CoHi06} for the
latest developments in the case of continuum random Schr\"odinger
operators. In its strongest form for lattice models, a Wegner estimate
provides Lipschitz continuity of the integrated density of states
$N(E)$.  In particular, this implies that the Lebesgue derivative of
$N(E)$, the density of states $n(E)$, exists as a function which is
essentially bounded from above. In addition to the upper bound for the
density of states, Wegner also presented an argument for a strictly
positive \emph{lower} bound for the density of states of the Anderson
model in his original article \cite{Weg81}.  Although insightful,
Wegner's argument is not complete as his nonzero lower bound
vanishes in the macroscopic limit.

In this note, we give a mathematical proof of a positive lower bound
for the density of states of the Anderson model. For many years,
efforts have been concentrated on Wegner's upper bound because this is
essential for the continuity of the integrated density of states, the
existence of the density of states, and for Anderson localization.
Consequently, Wegner's idea to obtain a lower bound seems to have
remained fairly unnoticed.  The lower bound is, however, essential for
Minami's proof that the energy level statistics for energies in the
strong localization regime is Poissonian \cite{minami1}.  Minami fixes
an energy $E$ in the region of complete localization at which the
fractional moment bounds of Aizenman and Molchanov \cite{AM1} hold.
He assumes that $n(E) > 0$. Minami then proves that the rescaled local
eigenvalue level spacing measure $\d\mu_L (x) = \sum_j \delta \bigl(
L^d (\varepsilon_j(L) - E) - x \bigr)\,\d x$, for the finite-volume
Hamiltonian (see below for the definition) with eigenvalues
$\varepsilon_j(L)$, converges in the macroscopic limit to a Poisson
distribution with density given by $n(E)$.  Here, we prove the
positivity of the density of states at almost every energy in the
deterministic spectrum. We mention that Molchanov \cite{molchanov1}
studied the same question of energy-level statistics for the
one-dimensional Russian school model. In his paper, he also proves the
positivity of the density of states for that model using completely
different methods.

The Anderson model is given by the discrete random Schr\"odinger
operator $H$ on a probability space $(\Omega,\mathbb{P})$ whose
realizations $H^{(\omega)} := \mathcal{L} + V^{(\omega)}$,
$\omega\in\Omega$, act as
\begin{equation}
  \label{H}
  (H^{(\omega)}\varphi)(x) = (\mathcal{L} \varphi)(x) + \omega_{x} \varphi(x)
\end{equation}
for all $x\in\mathbb{Z}^{d}$ on a dense domain of $\varphi
\in\ell^{2}(\mathbb{Z}^{d})$. Here, the discrete Laplacian
$\mathcal{L}$ is defined as $(\mathcal{L} \varphi)(x) :=
\sum_{y\in\mathbb{Z}^{d}: {|x-y|=1}} \varphi(y)$, and has purely
absolutely continuous spectrum $\sigma(\mathcal{L}) = [ -2d, 2d]$. The
random potential $V^{(\omega)}$ consists of a family
$\{\omega_{x}\}_{x\in\mathbb{Z}^{d}}$ of independent, identically
distributed real-valued random variables on $\Omega$.

The Schr\"odinger operator $H$ of the Anderson model is known
\cite{CaLa90, PaFi92} to be almost surely essentially self-adjoint on
the dense subspace $\{ \varphi\in\ell^{2}(\mathbb{Z}^{d}):
\supp\varphi \text{~~compact}\}$. Moreover, $H$ is ergodic with
respect to lattice translations. To define the integrated density of
states, we consider finite volumes $\Lambda \subset \Z^d$ and the
Dirichlet restriction $H_{\Lambda}^{(\omega)} := \mathcal{L}_{\Lambda}
+ V^{(\omega)}_\Lambda$ of $H$ to the finite-dimensional Hilbert space
$\ell^{2}(\Lambda)$, where $V^{(\omega)}_\Lambda$ is the restriction
of $V^{( \omega)}$ to $\Lambda$ and
\begin{equation}
  \label{LL}
  (\mathcal{L}_{\Lambda}\varphi)(x) :=  \sum_{y\in\Lambda: {|x-y|=1}}
  \varphi(y) +  \varphi(x) \Biggl( \sum_{y\notin\Lambda: {|x-y|=1}} 1 \Biggr)
\end{equation}
for all $x\in\Lambda$ and all $\varphi \in \ell^{2}(\Lambda)$.  Note
that the rightmost term in the above definition of the Dirichlet
Laplacian $\mathcal{L}_{\Lambda}$ ensures the Dirichlet-decoupling
estimate $\mathcal{L}_{\Lambda_{1} \cup \Lambda_{2}} \le
\mathcal{L}_{\Lambda_{1}} \oplus \mathcal{L}_{\Lambda_{2}}$, see also
\cite{Sim85, KiMu06}. We write $\tr_{\Lambda}$ for the trace on
$\ell^{2}(\Lambda)$ and let $\Chi_{B}$ stand for the indicator
function of some set $B \subset \mathbb{R}$. Then ergodicity implies
that the \emph{integrated density of states} $E \in \mathbb{R} \mapsto
N(E)$ is given by the non-random limit
\begin{equation}
  \label{Ndef}
  N(E) = \lim_{\Lambda \uparrow\mathbb{Z}^{d}} \left[ \frac{1}{|\Lambda|}
    \tr_{\Lambda} \left(\Chi_{]-\infty,  E]} (H_{\Lambda}^{(\omega)})
    \right)  \right]
\end{equation}
along a sequence of expanding cubes $\Lambda \subset\mathbb{Z}^{d}$.
Equation~\eqref{Ndef} holds for all $E\in\mathbb{R}$ that are continuity
points of $N(E)$ and all $\omega\in\Omega$ except for a $\mathbb{P}$-null set,
which can be chosen uniformly with respect to the aforementioned values of
$E$.

If the single-site distribution of, say, $\omega_{0}$ happens to be absolutely
continuous with respect to Lebesgue measure and if its Lebesgue density $\rho$
satisfies the additional assumption
\begin{equation}
  \label{rho-above}
  \rho_{\mathrm{max}} := \esssup_{ w\in \mathbb{R}}\; \{\rho(w)\} < \infty,
\end{equation}
then, according to Wegner \cite{Weg81}, the integrated density of states $N(E)$
is Lipschitz continuous, hence absolutely continuous and the Lebesgue
derivative of $N(E)$, the \emph{density of states},
\begin{equation}
  \label{ndef}
  E \in \mathbb{R}  \mapsto n(E) := \d N(E)/\d E,
\end{equation}
obeys the estimate $n(E) \le \rho_{\mathrm{max}}$ for
Lebesgue-almost all $E\in\mathbb{R}$.
Another consequence of ergodicty of the Schr\"odinger operator
$H$ is that there is a closed set $\Sigma \subset \mathbb{R}$
such that $\sigma (H^{(\omega)} ) = \Sigma$ with probability one. This set,
called the {\em deterministic spectrum} of $H$, is given by
$\Sigma = [ -2d, 2d] + \supp \rho$, for the model in (\ref{H}).

We will not assume \eqref{rho-above} for the validity of the lower
bound for $n(E)$.

\begin{theorem}
  \label{main}
  Let $H$ be the random Schr\"odinger operator \eqref{H} of the
  Anderson model. Assume that the single-site distribution of
  $\omega_{0}$ is absolutely continuous with respect to Lebesgue
  measure and that its Lebesgue density $\rho$ is essentially bounded away
  from
  zero on some interval $[W_-, W_+]$ in the sense that
  \begin{equation}
    \label{rho-below}
    \rho_{\mathrm{min}} := \essinf\limits_{ w\in [W_{-},W_{+}]} \;
    \{\rho(w)\} > 0
  \end{equation}
  for some $-\infty < W_{-} < W_{+} < \infty$. Assume further that the
  integrated density of states $N(E)$ is an absolutely continuous
  function with Lebesgue derivative $n(E)$ as in \eqref{ndef}.
  Then, for every $\delta >0$
  (small enough) there exists a strictly positive constant
  $C_{\delta}$ such that
  \begin{equation}
    \label{main-eq}
    n(E) \ge C_{\delta} > 0,
  \end{equation}
  for Lebesgue-almost all $E \in [-2d + W_{-} +\delta, 2d + W_{+}
  -\delta]$.
\end{theorem}

\begin{remark}
  The theorem can be generalized in a straightforward manner to
  incorporate general bounded, self-adjoint and $\mathbb{Z}^d$-translation
  invariant unperturbed
  operators $H_{0}$ instead of $\mathcal{L}$.  In this case,
  \eqref{main-eq} holds for all $E \in \sigma(H_{0}) + [W_{-} +\delta,
  W_{+} -\delta]$.
\end{remark}

\begin{remark}
  The lower bound constant $C_\delta > 0$ can be expressed in terms of the
  integrated density of states $N^{(0)} (E)$ for the unperturbed operator
  $\mathcal{L}$. 
  For any $\delta > 0$ small enough, we cover the interval $[-2d + W_{-}
  +\delta, 2d + W_+ - \delta ]$ by finitely many intervals $I_j = [ E_j -
  \delta , E_j + \delta]$ of width $2 \delta$ and centered at $E_j$.
  We can take $C_\delta$ to be 
  \begin{equation}
    \label{lowerbd1} 
    C_\delta = \min_{j} \left\{ \frac{ ( \delta \rho_{\mathrm{min}}
        )^{\alpha_{E_{j}}} }{2 \delta} \left[ N^{(0)} ( E_j - W_- - 2 \delta) -
        N^{(0)} ( E_j - W_+ + 2 \delta ) \right] \right\}, 
  \end{equation}
  where the positive constants $\alpha_{E_{j}} > 0$ are defined in the proof of
  Theorem~\ref{main}.  The difference of the integrated densities of states
  for $\mathcal{L}$ on the right of (\ref{lowerbd1}) is strictly positive, see
  the end of the proof of Theorem~\ref{main}.
\end{remark}

\noindent
In Theorem~\ref{main}, we do not require that $\rho$ is essentially
bounded as in (\ref{rho-above}), that the support is bounded, nor that
$[W_- , W_+]$ is the entire support of $\rho$.  However, if we add the
latter two hypotheses, we obtain the following special case of
Theorem~\ref{main}.

\begin{corollary}
  Under the hypotheses of Theorem~\ref{main}, and the additional condition
  that \eqref{rho-below} holds on the entire support of $\rho$, i.e.\
  $\rho(w)=0$ for almost every $w\in \mathbb{R} \setminus [W_{-}, W_{+}]$,
  then the density of states is strictly positive Lebesgue-almost everywhere
  on $\Sigma$.
\end{corollary}

\begin{remark}
  After we completed this work, we learned that a similar result is contained
  in Frank Jeske's unpublished PhD-thesis \cite{Jes92}, which was supervised
  by Werner Kirsch.  We thank Ivan Veseli\'c for informing us about the
  existence of \cite{Jes92}.
\end{remark}

\begin{remark}
  The question arises naturally whether a similar lower bound for the
  density of states $n$ does also hold in the case of \emph{continuum}
  random Schr\"odinger operators, that is, Schr\"odinger operators on $L^2
  (\mathbb{R}^d)$. For $d=1$ and for
  alloy-type random potentials with suitably well-behaved single-site
  potentials, the answer is affirmative. The argument proceeds as in
  the discrete case with some obvious modifications that are well
  known from proofs of upper Wegner estimates for continuum models.
  The key point is that the finite-rank-perturbation argument, which
  allows us to proceed from \eqref{Gint} to \eqref{Gfinal} below is
  still valid in the one-dimensional continuum case. Indeed, suppose
  we have two Schr\"odinger operators on an interval that differ only
  by a boundary condition (Dirichlet versus none, say) that is imposed
  at an \emph{interior} point of the interval. Then it is well known
  from, e.g., the theory of point interactions that these two
  Schr\"odinger operators differ by a rank-2 perturbation.  However,
  for $d\ge 2$, different boundary conditions along a
  finite hypersurface $S$ give rise to an infinite-rank perturbation.
  Thus, in the case $d\ge 2$ one needs an alternative argument why
  different boundary conditions along $S$ for Schr\"odinger operators
  in a finite volume $\Lambda$ (with $S$ in the interior of $\Lambda$)
  lead to eigenvalue counting functions that differ by a term
  proportional to the area of $S$. Furthermore, this error term would
  be required to remain bounded as $\Lambda\uparrow\mathbb{R}^{d}$.
  But this is a delicate issue in view of \cite{Kir87,Kir89}.
\end{remark}

\begin{proof}[Proof of Theorem~\ref{main}]
  The proof follows Wegner's arguments \cite{Weg81}, except that we introduce
  a partition of the finite volume into cubes of large but fixed size. This
  allows us to get a nontrivial result in the macroscopic limit $L \rightarrow
  \infty$, a problem seemingly ignored in \cite{Weg81}.  A similar
  partitioning strategy was used in the proof of a (upper) Wegner estimate for
  continuum random Schr\"odinger operators by spectral averaging \cite{CH},
  see also \cite{FiHu97b, HuLe01a} for the case of Gaussian or other types of
  unbounded random potentials.

  \noindent  
  \textbf{1.} \quad  
  Let $E_{1}, E_{2} \in \mathbb{R}$ such that $E_{2} - E_{1} > \varepsilon$
  for some $\varepsilon >0$. We consider a sequence of expanding cubes
  $\Lambda_{L}$ in $\mathbb{Z}^{d}$ with volume $|\Lambda_{L}|=L^{d}$.
  Finally, we pick a smooth, monotone increasing switch function
  $f_{\varepsilon} \in C^{1}(\mathbb{R})$ such that $f_{\varepsilon}(\lambda)
  = 0$ for all $\lambda \le 0$ and $f_{\varepsilon}(\lambda) =1$ for all
  $\lambda \ge \varepsilon$.  We let $\mathbb{E}$ denote the expectation
  associated with the probability measure $\mathbb{P}$, and we write
  $F_{\varepsilon,L}(\lambda, \omega ) := \tr_{\Lambda_{L}}
  f_{\varepsilon}(\lambda - H^{(\omega)}_{\Lambda_{L}})$. Then we have
  \begin{align}\label{switch0}
    N(E_{2}) - N(E_{1}) & \ge \lim_{L\to\infty}  \left\{
     \frac{1}{L^{d}} \; \mathbb{E} \Bigl[ \tr_{\Lambda_{L}} \bigl(
       f_{\varepsilon}(E_{2} - H_{\Lambda_{L}}) -
       f_{\varepsilon}(E_{1} + \varepsilon
       - H_{\Lambda_{L}})
       \bigr)\Bigr] \right\} \nonumber\\
    & =  \lim_{L\to\infty} \left\{
     \frac{1}{L^{d}} \;
    \int_{E_{1}+\varepsilon}^{E_{2}}\! \d\lambda \; \mathbb{E}\left[
      \frac{\partial}{\partial\lambda}\; F_{\varepsilon,L}(\lambda,
      \cdot) \right] \right\}.
  \end{align}
  The quantity
  $F_{\varepsilon,L}(\lambda, \omega )$
   depends on
  $\lambda$ and $\omega$ only through the differences $\{\omega_{x}
  -\lambda\}_{x\in\Lambda_{L}}$, and it is a monotone decreasing function
  in each of those differences.
We partition the cube $\Lambda_{L}$ into $(L/\ell)^{d}$ smaller cubes
$\Gamma_{j}$ of the same (fixed) volume $\ell^{d}$. We consider only
those big cubes $\Lambda_{L}$ for which such a partition is possible.
We will take $L \rightarrow \infty$, and $\ell$ large but finite.
Therefore we get
  \begin{equation}
    \label{switch}
    \frac{\partial}{\partial\lambda} \; F_{\varepsilon,L}(\lambda, \omega)
    = - \sum_{j=1}^{(L/\ell)^{d}} \sum_{x\in\Gamma_{j}}
    \frac{\partial}{\partial\omega_{x}} \;
    F_{\varepsilon,L}(\lambda,\omega)
  \end{equation}
  for all $\lambda\in\mathbb{R}$ and all $\omega\in\Omega$.

  \noindent  
  \textbf{2.} \quad  
  We conclude from \eqref{switch0} and \eqref{switch} that
  \begin{equation}
    \label{Nstart}
    N(E_{2}) -N(E_{1}) \ge \rho_{\mathrm{min}}^{\ell^{d}} \lim_{L\to\infty}
    \left\{
    \frac{1}{(L/\ell)^{d}} \sum_{j=1}^{(L/\ell)^{d}}
    \mathbb{E}_{\Gamma_{j}^{c}} \left[ \int_{E_{1}+\varepsilon}^{E_{2}}\!
      \d\lambda \; G_{j}(\lambda,\cdot)\right] \right\}
  \end{equation}
  with
  \begin{equation}
    \label{Gdef}
    G_{j}(\lambda,\omega_{\Gamma_{j}^{c}}) := \frac{1}{\ell^{d}}
    \int_{[W_{-},W_{+}]^{\ell^{d}}}
    \! \Bigl(\prod_{y\in\Gamma_{j}}\d\omega_{y} \Bigr) \sum_{x\in\Gamma_{j}}
    \Bigl(- \frac{\partial}{\partial\omega_{x}} \Bigr)
    F_{\varepsilon,L}(\lambda,\omega).
  \end{equation}
  Here $\Gamma_{j}^{c} := \mathbb{Z}^{d}\setminus \Gamma_{j}$ denotes the
  complement of $\Gamma_{j}$, and (in slight abuse of notation) we have
  written $\omega =: (\omega_{\Gamma_{j}}, \omega_{\Gamma_{j}^{c}})$, where
  $\omega_{\Gamma_{j}} := (\omega_{x})_{x\in\Gamma_{j}}$. The partial disorder
  average $\mathbb{E}_{\Gamma_{j}^{c}}$ in \eqref{Gdef} extends only over the
  coupling constants $\omega_{\Gamma_{j}^{c}}$.

  \noindent  
  \textbf{3.} \quad  
  Following Wegner \cite{Weg81}, we are going to perform a change of variables
  in \eqref{Gdef} from $\omega_{\Gamma_{j}}$ to $\eta$: we fix an arbitrary
  point $x_{j}\in\Gamma_{j}$ and set $\eta_{x_{j}} := \omega_{x_{j}}$ and
  $\eta_{y} := \omega_{y} - \omega_{x_{j}}$ for all $y\in\Gamma_{j}\setminus
  \{x_{j}\}$. The Jacobian associated with this change of variables is $1$,
  whence
  \begin{align}
    \label{Gchange}
    G_{j}(\lambda,\omega_{\Gamma_{j}^{c}}) =  \frac{1}{\ell^{d}}
    \int_{[W_{-},W_{+}]} \!\d\eta_{x_{j}} & \int_{[W_{-}-\eta_{x_{j}}, W_{+}
      -\eta_{x_{j}}]^{\ell^{d} -1}} \!
    \Bigl(\prod_{y\in\Gamma_{j}\setminus\{x_{j}\}} \d\eta_{y} \Bigr)
    \nonumber\\
    & \qquad\times \Bigl(- \frac{\partial}{\partial\eta_{x_{j}}} \Bigr)
    F_{\varepsilon,L}\bigl(\lambda,(\omega_{\Gamma_{j}}(\eta),
    \omega_{\Gamma_{j}^{c}})\bigr) .
  \end{align}
  Now, fix $\delta \in ]0, (W_{+} -W_{-})/4[$.  One obtains a lower bound for
  \eqref{Gchange} by restricting first the integration over $\eta_{x_{j}}$ to
  $[W_{-} + \delta/2, W_{+} -\delta/2]$ and then restricting the integration
  over $\eta_{y}$ to $[-\delta/2,\delta/2]$, for all
  $y\in\Gamma_{j}\setminus\{x_{j}\}$. This gives
  \begin{align}
    \label{Ghalfint}
    G_{j}(\lambda,\omega_{\Gamma_{j}^{c}}) \ge  \frac{1}{\ell^{d}}
    \int_{[-\delta/2,\delta/2]^{\ell^{d} -1}} \!
    \Bigl(\prod_{y\in\Gamma_{j}\setminus\{x_{j}\}} \d\eta_{y} \Bigr)
    \Bigl[ & F_{\varepsilon,L}\bigl(\lambda,(\omega_{\Gamma_{j}}(\eta^{-}),
    \omega_{\Gamma_{j}^{c}}) \bigr)\nonumber \\
    & - F_{\varepsilon,L}\bigl(\lambda,(\omega_{\Gamma_{j}}(\eta^{+}),
    \omega_{\Gamma_{j}^{c}})\bigr) \Bigr]
  \end{align}
  with $\eta^{\pm} := \bigl( W_{\pm} \mp \delta/2,
  (\eta_{y})_{y\in\Gamma_{j}\setminus\{x_{j}\}} \bigr)$. Note that in
  \eqref{Ghalfint} one has
  $\bigl(\omega_{\Gamma_{j}}(\eta^{-})\bigr)_{x} \le W_{-} +\delta$ and
  $\bigl(\omega_{\Gamma_{j}}(\eta^{+})\bigr)_{x} \ge W_{+} -\delta$ for
  all $x\in\Gamma_{j}$.
  Since $F_{\varepsilon,L}$ is a decreasing function in each
  $\omega_{x}$, we arrive at
  \begin{equation}
    \label{Gint}
    G_{j}(\lambda,\omega_{\Gamma_{j}^{c}}) \ge
    \frac{\delta^{\ell^{d}-1}}{\ell^{d}}
    \Bigl[ F_{\varepsilon,L}\bigl(\lambda,(\omega_{\Gamma_{j}}^{-},
    \omega_{\Gamma_{j}^{c}}) \bigr)
    - F_{\varepsilon,L}\bigl(\lambda,(\omega_{\Gamma_{j}}^{+},
    \omega_{\Gamma_{j}^{c}})\bigr)   \Bigr]
  \end{equation}
  with spatially constant couplings $\omega_{\Gamma_{j}}^{\pm} :=
  (W_{\pm} \mp \delta)_{x\in\Gamma_{j}}$ inside the small cube
  $\Gamma_{j}$.

  \noindent  
  \textbf{4.} \quad  
  Next, we will use a Dirichlet decoupling of the small cube
  $\Gamma_{j}$. In the first (i.e.\ the positive) term on the
  right-hand side of \eqref{Gint}, this can be done straight away,
  because $H_{\Lambda_{L}} \le H_{\Gamma_{j}} \oplus H_{\Lambda_{L}
    \setminus\Gamma_{j}}$. To do the replacement in the second (i.e.\
  the negative) term, one has to take into account the error that
  arises from introducing the additional Dirichlet boundary condition
  along $\partial\Gamma_{j} \setminus \partial \Lambda_{L}$. But this
  is a perturbation of rank $\mathcal{O}(\ell^{d-1})$ and it is independent
  of the coupling constants. Furthermore, recall that $0 \le
  f_{\varepsilon} \le 1$. Thus there is a constant $D\in
  ]0,\infty[$, which depends only on $d$, such that
  \begin{align}
    \label{Gfinal}
    G_{j}(\lambda,\omega_{\Gamma_{j}^{c}}) & \ge
    \frac{\delta^{\ell^{d}-1}}{\ell^{d}} \; \Bigl\{
      \tr_{\Gamma_{j}} \bigl[ f_{\varepsilon}(\lambda -
      W_{-} - \delta  - \mathcal{L}_{\Gamma_{j}})
      -  f_{\varepsilon}(\lambda -
      W_{+} + \delta  - \mathcal{L}_{\Gamma_{j}})
      \bigr] \nonumber\\
      & \hspace*{1.7cm} - D \ell^{d-1} \Bigr\} \nonumber\\
      & \ge
      \frac{\delta^{\ell^{d}-1}}{\ell^{d}} \; \Bigl\{
      \tr_{\Gamma_{j}} \bigl[ f_{\varepsilon}(E_{1} -
      W_{-} - \delta  - \mathcal{L}_{\Gamma_{j}})
      -  f_{\varepsilon}(E_{2} -
      W_{+} + \delta  - \mathcal{L}_{\Gamma_{j}})
      \bigr] \nonumber\\
      & \hspace*{1.7cm} - D \ell^{d-1} \Bigr\}
  \end{align}
  for all $\lambda \in [E_{1}, E_{2}]$ and all $\omega_{\Gamma_{j}^{c}}$.
  The contributions from $H_{\Lambda_{L}\setminus\Gamma_{j}}$ have
canceled, so the right side in (\ref{Gfinal}) is
independent of $L$.
  Inserting \eqref{Gfinal} into \eqref{Nstart} and taking the limit
  $\varepsilon\downarrow 0$, we arrive at the
  estimate
  \begin{equation}
    \label{diffquot}
    \frac{N(E_{2}) -N(E_{1})}{E_{2} -E_{1}} \ge
    \frac{(\delta\,\rho_{\mathrm{min}})^{\ell^{d}}}{\delta}
    \bigl( K_{\ell}(E_{1},E_{2}) -D/\ell\bigr)
  \end{equation}
  for the difference quotient of the integrated density of states of $H$.
  The lower bound in \eqref{diffquot}
  is expressed in terms of the difference
  \begin{equation}
    K_{\ell}(E_{1},E_{2}) :=
    N_{\Lambda_{\ell}}^{(0)}(E_{1} -W_{-} -\delta) -
    N_{\Lambda_{\ell}}^{(0)}(E_{2} - W_{+} +\delta)
  \end{equation}
  of the  Dirichlet finite-volume approximation
  $N_{\Lambda_{\ell}}^{(0)}(\lambda) := \ell^{-d} \tr_{\Lambda_{\ell}}
  \Chi_{]-\infty,\lambda]}(\mathcal{L}_{\Lambda_{\ell}}) $ for the integrated
  density of states of the free Laplacian $\mathcal{L}$. 

  \noindent 
  \textbf{5.} \quad  
  By hypothesis we know that $N(E)$ is
  absolutely continuous with respect to Lebesgue measure. Hence, we can take
  the monotone limit $E_{2}\downarrow E_{1} =:E$ in \eqref{diffquot} and
  obtain
  \begin{equation}
    \label{almostdone}
    n(E) \ge \frac{(\delta\,\rho_{\mathrm{min}})^{\ell^{d}}}{\delta}
    \bigl( K_{\ell}(E,E) -D/\ell\bigr)
  \end{equation}
  for Lebesgue-almost all $E\in\mathbb{R}$.  We now fix $E_0 \in \R$,
  and observe that $K_{\ell}(E,E) \ge K_{\ell}(E_{0}-\delta,
  E_{0}+\delta) =: K_{\ell}(E_{0})$, for all $E\in
  [E_{0}-\delta,E_{0}+\delta]$.  We next note that 
  \begin{equation}
    \label{limit1}
    K(E_{0}) := \lim_{\ell \to\infty} K_{\ell}(E_{0}) = N^{(0)}(E_{0}-
    W_{-} -2\delta) - N^{(0)}(E_{0} -W_{+} +2\delta) 
  \end{equation} 
  exists, where $N^{(0)}(\lambda) := \lim_{\ell\to\infty}
  N_{\Lambda_{\ell}}^{(0)}(\lambda) = \langle \delta_{0},
  \Chi_{]-\infty,\lambda]}(\mathcal{L}) \delta_{0} \rangle$.  It is
  important to observe that $E_{0}- W_{-}-2\delta > E_{0} -W_{+}
  +2\delta$, since $0 < \delta < (W_+ - W_- ) / 4$, and that for all
  $E_{0} \in ]-2d + W_{-}+2\delta, 2d + W_{+} -2\delta[$, we have $-2d
  < E_0- W_- -2\delta < 2d + (W_+ - W_-) - 4 \delta$, and $-2d - [
  (W_+ - W_-) - 4 \delta] < E_0 -W_+ +2\delta < 2d$.  Specifically, we
  have $-2d < E_0- W_- -2\delta$ and if $E_{0}- W_{-}-2\delta > 2d$,
  then the other energy satisfies $E_{0}-W_{+} +2\delta < 2d$.
  Consequently, $K (E_0)$ is strictly positive for $E_0$ on the
  specified range since
  \begin{equation}
    \label{unperturbed}
    N^{(0)}(\lambda_{2}) -N^{(0)}(\lambda_{1}) >0,
  \end{equation}
  whenever $\lambda_{1} < \lambda_{2}$, and at least one of the
  $\lambda_{j}$'s lies in the interior of $\sigma(\mathcal{L}) = [-2d,2d]$.
  Thus, there exists a finite length $\ell_{E_{0}}$ such that
  \begin{equation}
    \label{freebound}
    K_{\ell_{E_{0}}}(E,E) - D/\ell_{E_{0}} \ge
    K(E_{0})/2 \qquad \text{for all}\quad
    E\in [E_{0}-\delta,E_{0}+\delta].
  \end{equation}
  The theorem follows from \eqref{almostdone}, \eqref{freebound} and by
  covering the interval $] -2d + W_{-}+\delta, 2d+ W_{+} - \delta[$ by a
  finite number of small intervals of length $2\delta$.
\end{proof}


\begin{acknowledgement}
  The authors thank Fran\c{c}ois Germinet for the kind hospitality at the
  Universit\'e de Cergy-Pontoise, and Abel Klein for several discussions.
\end{acknowledgement}



\begin{thebibliography}{AENSS}
\frenchspacing

\bibitem[AM]{AM1} 
  \au{M.}{Aizenman}\et\lau{S.}{Molchanov}
  \ti{Localization at large disorder and extreme energies: An elementary
    derivation} 
  \z{Commun. Math. Phys.}{157}{245--278}{1993}

\bibitem[CL]{CaLa90}
  \au{R.}{Carmona}\et\lau{J.}{Lacroix}
  \bti{Spectral theory of random Schr\"odinger operators}
  \pub{Birkh\"auser}{Boston}{1990}

\bibitem[CH]{CH}
  \au{J.-M.}{Combes}\et\lau{P. D.}{Hislop}
  \ti{Localization for some continuous, random Hamiltonian in
    $d$-dimensions}
  \z{J. Funct. Anal.}{124}{149--180}{1994}

\bibitem[CHK]{CoHi06}
  \au{J.-M.}{Combes}, \au{P. D.}{Hislop}\et\lau{F.}{Klopp}
  \ti{An optimal Wegner estimate and its application to the global
    continuity of the integrated density of states for random
    Schr\"odinger operators}
  e-print arXiv:math-ph/0605029v2; to appear in Duke Math.\ J.\

\bibitem[FHLM]{FiHu97b}
  \au{W.}{Fischer}, \au{T.}{Hupfer}, \au{H.}{Leschke}, \lau{P.}{M{\"u}ller}
  \ti{Existence of the density of states for multi-dimensional
    continuum Schr{\"o}dinger operators with Gaussian random potentials}
  \z{Commun. Math. Phys.}{190}{133--141}{1997}

\bibitem[HLMW]{HuLe01a}
  \au{T.}{Hupfer}, \au{H.}{Leschke}, \au{P.}{M\"uller}, \lau{S.}{Warzel}
  \ti{The absolute continuity of the integrated density of states for
    magnetic Schr\"odinger operators with certain unbounded random
    potentials}
  \z{Commun. Math. Phys.}{221}{229--254}{2001}

\bibitem[J]{Jes92}
  \lau{F.}{Jeske}
  \bti{\"Uber lokale Positivit\"at der Zustandsdichte zuf\"alliger
    Schr\"odinger-Operatoren}
  PhD-thesis, Ruhr-Universit\"at Bochum, Germany, 1992 [in German].

\bibitem[K1]{Kir87}
  \lau{W.}{Kirsch}
  \ti{Small perturbations and the eigenvalues of the Laplacian on
    large bounded domains}
  \z{Proc. Amer. Math. Soc.}{101}{509--512}{1987}

\bibitem[K2]{Kir89}
  \lau{W.}{Kirsch}
  \ti{The stability of the density of states of Schr\"odinger operator
    under very small perturbations}
  \z{Integral Equations Operator Theory}{12}{383--391}{1989}

\bibitem[KMe]{KiMe07}
  \au{W.}{Kirsch}\et\lau{B.}{Metzger}
  \ti{The integrated density of states for random Schr\"odinger operators}
  In:
  \bti{Spectral Theory and Mathematical Physics: A Festschrift in
    Honor of Barry Simon's 60th Birthday, Part 2}
  Proc. Symp. Pure Math., vol. {\bf 76},
  \pub[649--696]{Amer. Math. Soc.}{Providence, RI}{2007}

\bibitem[KiM\"u]{KiMu06}
  \au{W.}{Kirsch}\et\lau{P.}{M\"uller}
  \ti{Spectral properties of the Laplacian on bond-percolation graphs}
  \z{Math. Z.}{252}{899--916}{2006}
  
\bibitem[Min]{minami1} 
  \lau{N.}{Minami}
  \ti{Local Fluctuation of the spectrum of a multidimensional Anderson tight
    binding model}
  \z{Commun. Math. Phys.}{177}{709--725}{1996}

\bibitem[Mol]{molchanov1} 
  \lau{S. A.}{Mol\v{c}anov}
  \ti{The local structure of the spectrum of the one-dimensional
    Schr\"odinger operator}
  \z{Commun. Math. Phys.}{78}{429--446}{1981}

\bibitem[PF]{PaFi92}
  \au{L.}{Pastur}\et\lau{A.}{Figotin}
  \bti{Spectra of random and almost-periodic operators}
  \pub{Springer}{Berlin}{1992}

\bibitem[S]{Sim85}
  \lau{B.}{Simon}
  \ti{Lifschitz tails for the Anderson model}
  \z{J. Stat. Phys.}{38}{65--76}{1985}

\bibitem[V]{Ves04}
  \lau{I.}{Veseli\'c}
  \ti{Integrated density of states and Wegner estimates for random
    Schr\"odinger operators}
  In:
  \bti{Spectral theory of Schr\"odinger operators}
  Contemp. Math., vol. {\bf 340},
  \pub[97--183]{Amer. Math. Soc.}{Providence, RI}{2004}

\bibitem[W]{Weg81}
  \lau{F.}{Wegner}
  \ti{Bounds on the density of states in disordered systems}
  \z{Z. Phys. B}{44}{9--15}{1981}

\end{thebibliography}
\end{document}